\documentclass[conference,letterpaper]{IEEEtran}

\usepackage{cite}
\usepackage[pdftex]{graphicx,color}
\usepackage[cmex10]{amsmath}

\usepackage{tikz,pgf,euscript,enumerate}

\newtheorem{theorem}{Theorem}

\begin{document}
\title{The Gaussian Two-way Diamond Channel}

\author{\IEEEauthorblockN{Prathyusha V, Srikrishna Bhashyam and Andrew Thangaraj}
\IEEEauthorblockA{Department of Electrical Engineering\\
Indian Institute of Technology Madras, India\\
Email: \{skrishna,andrew\}@ee.iitm.ac.in}
}
\maketitle

\begin{abstract}
  \let\thefootnote\relax\footnotetext{This work was done at the
    Department of Electrical Engineering, IIT Madras,
    Chennai. India. Prathyusha V is now with Cisco Systems, Bangalore,
    India. Srikrishna Bhashyam and Andrew Thangaraj are with the
    Department of Electrical Engineering, IIT Madras, Chennai,
    India. }  We consider two-way relaying in a Gaussian diamond
  channel, where two terminal nodes wish to exchange information using
  two relays. A simple baseline protocol is obtained by time-sharing
  between two one-way protocols. To improve upon the baseline
  performance, we propose two compute-and-forward (CF) protocols --
  Compute-and-forward-Compound multiple access channel (CF-CMAC) and
  Compute-and-forward-Broadcast (CF-BC). These protocols mix the two
  flows through the two relays and achieve rates better than the
  simple time-sharing protocol. We derive an outer bound to the
  capacity region that is satisfied by any relaying protocol, and
  observe that the proposed protocols provide rates close to the outer
  bound in certain channel conditions. Both the CF-CMAC and CF-BC
  protocols use nested lattice codes in the compute phases. In the
  CF-CMAC protocol, both relays simultaneously forward to the
  destinations over a Compound Multiple Access Channel (CMAC). In the
  simpler CF-BC protocol's forward phase, one relay is selected at a
  time for Broadcast Channel (BC) transmission depending on the
  rate-pair to be achieved. We also consider the diamond channel with
  direct source-destination link and the diamond channel with
  interfering relays. Outer bounds and achievable rate regions are
  compared for these two channels as well. Mixing of flows using the
  CF-CMAC protocol is shown to be good for symmetric two-way rates.
\end{abstract}

\IEEEpeerreviewmaketitle

\section{Introduction}
The half-duplex Gaussian diamond channel where a source communicates
with a destination through two non-interfering relays is a problem of interest in information theory \cite{SchienDiamond}. In \cite{BagMotKha10}, multihopping
decode-and-forward (MDF) protocols were proposed to achieve rates
within a constant gap of a capacity outer bound.

Two-way relaying where two nodes without a direct link communicate
with each other through a single relay has been studied in
\cite{PopYom07,kim2008performance,kim2011achievable,NamChuLee10,WilNar10,NazerGastpar}. In
\cite{PopYom07,kim2008performance,kim2011achievable}, the achievable
rate regions of various two-way relaying protocols are compared. In
\cite{NamChuLee10,WilNar10,NazerGastpar}, coding strategies based on nested lattice
codes and compute-and-forward are proposed. One interesting aspect of
two-way relaying is that there are {\em two} data flows and {\em
  mixing of the two flows} at the relay can be exploited to improve
the rates in both directions. This is achieved by physical layer
network coding or the compute-and-forward strategy. Two-way relaying
where the two communicating nodes also have a direct link has been
studied in
\cite{kim2011achievable,GonYueWan11,tian2012asymmetric,allerton12}.
While protocols and achievable rate regions are proposed in
\cite{kim2011achievable,GonYueWan11,tian2012asymmetric}, an outer
bound to the capacity region for any protocol is derived in
\cite{allerton12}.

In this paper, we consider two-way communication over the diamond
channel, which appears to have not received attention in existing literature. In this network shown in Fig. \ref{diamond}, nodes $A$ and
$B$ communicate with each other through two non-interfering relays
$R_1$ and $R_2$.
\begin{figure*}[t]
\begin{center}
\begin{minipage}{1.8in}
\label{diamond}
\tikzstyle{cnode}=[circle,draw]
\begin{tikzpicture}
\node [cnode,left] (A) at (0,0) {A};
\node [cnode,right] (B) at (2.4,0) {B};
\node [cnode,above] (R1) at (1.2,1) {$R_1$};
\node [cnode,below] (R2) at (1.2,-1) {$R_2$};
\draw[-] (A) to node[left]{$\gamma_{a1}$} (R1);
\draw[-] (A) to node[left]{$\gamma_{a2}$} (R2);
\draw[-] (R2) to node[right]{$\gamma_{b2}$} (B);
\draw[-] (R1) to node[right]{$\gamma_{b1}$} (B);
\end{tikzpicture}
\caption{Gaussian Diamond Channel}
\end{minipage}
\begin{minipage}{4.0in}
\label{states}
\includegraphics[width=4.0in]{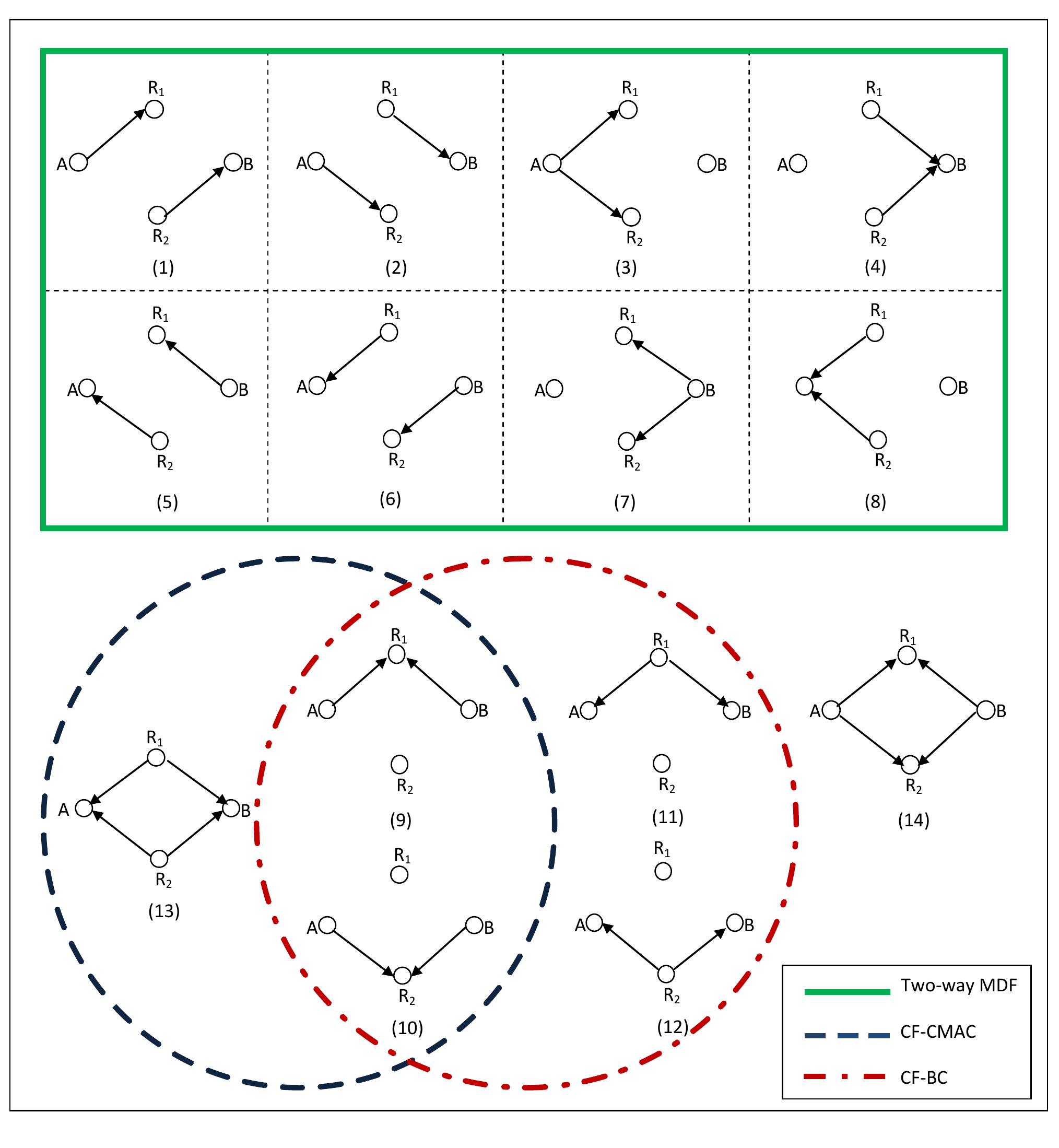}
\caption{States of a Diamond Channel}
\end{minipage}
\end{center}
\end{figure*}
We are interested in the capacity region consisting
of all possible rate pairs $(R_a, R_b)$, where $R_a$ is the rate of
communication from $A$ to $B$ and $R_b$ is the rate of communication
from $B$ to $A$. First, we derive an outer bound to the capacity
region that is valid for any protocol. This is obtained by extending
the approach in \cite{allerton12} to the diamond channel. Then, we
propose relaying protocols for the two-way diamond channel and
determine their achievable rate regions. 

A simple baseline protocol is
a two-way protocol that does not mix the two flows between the nodes $A$ and $B$, and time-shares
between two one-way MDF protocols \cite{BagMotKha10} in either direction. We call this
the two-way MDF protocol. Two compute-and-forward protocols --
Compute-and-forward-Compound multiple access channel (CF-CMAC) and
Compute-and-forward-Broadcast (CF-BC)-- that can achieve rate pairs
that the two-way MDF protocol cannot achieve are proposed. In the
CF-BC protocol, only one of the relays is used at any given time based
on the required rate-pair. This allows the use of compute-and-forward
schemes known for the one-relay case. In the CF-CMAC protocol, a
nested lattice code is used in the transmission to the relays and the
two relays forward to the two destinations simultaneously. The use of
the compound MAC in the forwarding phase of CF-CMAC protocol instead
of two separate broadcast phases from the two relays improves the
achievable rate region compared to the CF-BC protocol. Finally, we
consider the possibility of time-sharing between the CF-CMAC and
two-way MDF protocols to improve the achievable rate region. Numerical
results and comparisons for different channel conditions are shown to
illustrate the gains from the proposed protocols. 

We also consider the diamond channel with direct source-destination
link and the diamond channel with interfering relays. Outer bounds and
achievable rate regions are compared. We develop a 2-relay Cooperative
Multiple Access Broadcast Channel (2-relay CoMABC) protocol for the
diamond channel with direct link and compared the achievable rate
region with the CF-BC and CF-CMAC protocols and the outer bound. We
also compare the 2-way Alternating Relay DF (2-way AR-DF) protocol
with the CF-CMAC protocol and the outer bound for the diamond channel
with interfering relays.

\section{Two-way Gaussian Diamond Channel}
The Gaussian diamond channel is shown in Fig. \ref{diamond}
\cite{SchienDiamond,BagMotKha10}. Two-way communication between two nodes, $A$ and
$B$, is assisted by two relays, $R_1$ and $R_2$. No direct link is
assumed to be present between the two nodes $A$ and $B$ or between the
two relays. We consider half-duplex nodes, i.e, at any particular
time, a node can either be in transmit state or in receive state. The
links are Gaussian with receiver noise variance of $N$ and
reciprocal. Let $P$ be the transmit power available at each node. The
SNRs of the different links are denoted as
$\gamma_{a1}=\frac{{h}_{a1}^2P}{N}$,$\gamma_{a2}=\frac{{h}_{a2}^2P}{N}$,$\gamma_{b1}=\frac{{h}_{b1}^2P}{N}$
and $\gamma_{b2}=\frac{{h}_{b2}^2 P}{N}$, where
$h_{a1},h_{a2},h_{b1},h_{b2}$ are the gains of the links $A\leftrightarrow{R_1},
A\leftrightarrow{R_2}, B\leftrightarrow{R_1}, B\leftrightarrow{R_2}$, respectively. We use
$C(\gamma)={\log_2}(1+\gamma)$ to represent the capacity of a
circularly symmetric complex Gaussian channel with SNR $\gamma$.

The diamond relay network has $2^4 = 16$ possible states, since each of
the four half-duplex nodes can either be in transmit or receive
state. Of these 16 states, we can ignore the 2 states in which all
nodes are in transmit or all are in receive states, as they do not
help in information flow. The 14 useful states are shown in
Fig. \ref{states}. As an illustration, when state 14 is
used, nodes $A$ and $B$ transmit symbols denoted $X_A$ and $X_B$, and
this is received as $Y_{R_1}=h_{a1}X_A+h_{b1}X_B+Z_{R_1}$ at $R_1$ and
$Y_{R_2}=h_{a2}X_A+h_{b2}X_B+Z_{R_2}$ at $R_2$ with $Z$ denoting additive
noise. Specifying a relaying protocol involves specifying
the sequence of states and the coding/decoding schemes for each of
these states. In one-way communication over the diamond channel in
\cite{BagMotKha10}, only states 1 to 4 need to be considered for
communication from $A$ to $B$. In two-way communication, all 14 states
need to be considered in general. 

As mentioned before, $R_a$ is the rate of communication from $A$ to
$B$ and $R_b$ is the rate of communication from $B$ to $A$. The
capacity region for two-way communication consists of all $(R_a, R_b)$
pairs for which reliable two-way communication is possible.

\section{Capacity Region Outer Bound}
\label{outer-diamond}
In this section, we derive an outer bound to the capacity region of
the two-way Gaussian diamond channel, which holds for any two-way
relaying protocol.  The bound is derived using the half-duplex cutset
bound \cite{khojastepour2003capacity} and the boundary of this region
can be computed by solving linear programs.

In general, we should consider all 14 states for deriving the outer
bound. However, if a state has all cut capacities lesser than or equal
to the corresponding cut capacities for another state, then only the
state with higher cut capacities needs to be considered. In this
network, 8 of the states (3,4 and 7-12) are dominated by either State
13 or 14. Therefore, we consider only the 6 states: 1, 2, 5,
6, 13, 14 with the fraction of time the network in state $i$ being
denoted $\mu_i$.

\begin{theorem}
\label{thm1}
Any achievable rate pair $(R_a, R_b)$ for two-way communication over the diamond channel must satisfy the following inequalities for some $\{\mu_i\}$: 
{\allowdisplaybreaks\begin{align}
R_a &\leq {\mu_{14}}C(\gamma_{a1} + \gamma_{a2}) + {\mu_1}C(\gamma_{a1}) +  {\mu_2}C(\gamma_{a2}), \nonumber \\
R_a &\leq {\mu_{14}}C(\gamma_{a2}) + {\mu_{13}}C(\gamma_{b1}) +  {\mu_2}(C(\gamma_{a2})+C(\gamma_{b1})), \nonumber \\
R_a &\leq {\mu_{14}}C(\gamma_{a1}) + {\mu_{13}}C(\gamma_{b2}) +  {\mu_1}(C(\gamma_{a1})+C(\gamma_{b2})), \nonumber \\
R_a &\leq {\mu_{13}}C((\sqrt{\gamma_{b1}}+\sqrt{\gamma_{b2}})^2) + {\mu_1}C(\gamma_{b2}) +  {\mu_2}C(\gamma_{b1}), \nonumber \\
R_b &\leq {\mu_{14}}C(\gamma_{b1} + \gamma_{b2}) + {\mu_5}C(\gamma_{b1}) +  {\mu_6}C(\gamma_{b2}), \nonumber \\
R_b &\leq {\mu_{14}}C(\gamma_{b2}) + {\mu_{13}}C(\gamma_{a1}) +  {\mu_6}(C(\gamma_{a1})+C(\gamma_{b2})), \nonumber \\
R_b &\leq {\mu_{14}}C(\gamma_{b1}) + {\mu_{13}}C(\gamma_{a2}) +  {\mu_5}(C(\gamma_{b1})+C(\gamma_{a2})), \nonumber \\
R_b &\leq {\mu_{13}}C((\sqrt{\gamma_{a1}}+\sqrt{\gamma_{a2}})^2) + {\mu_5}C(\gamma_{a2}) +  {\mu_6}C(\gamma_{a1}), \nonumber \\
\sum_{i} & \mu_i = 1, \mu_i\ge0.
\label{outerbound}
\end{align}}
\begin{proof}
For any general network with $M$ states and a fraction of time $\mu_i$ in state $i$, any achievable rate $R$ of information flow is bounded as
\begin{equation}
R \leq {\min_S} {\sum_{i=1}^M}{\mu_i I(X^S;Y^{S^c}|X^{S^c},i)},
\end{equation}
where $R$ is the rate from source to destination node with the source
in a subset of nodes $S$ and the destination in
$S^c$. The set $S$ defines a cut that separates source and
destination. This bound is applied to the two-way diamond channel using the
cuts
$\lbrace{a}\rbrace$,$\lbrace{a,r_1}\rbrace$,$\lbrace{a,r_2}\rbrace$,$\lbrace{a,r_1,r_2}\rbrace$
for bounding $R_a$ and the cuts
$\lbrace{b}\rbrace$,$\lbrace{b,r_1}\rbrace$,$\lbrace{b,r_2}\rbrace$,$\lbrace{b,r_1,r_2}\rbrace$
for bounding $R_b$. Using these cuts, we obtain the following.
\begin{equation*}
R_a \leq  \min{\lbrace R_{a1},R_{a2},R_{a3},R_{a4} \rbrace}, 
\end{equation*}
\begin{equation*}
\begin{split}
\mbox{where~~} R_{a1} = &{\mu_{14}}I(X_a;Y_1,Y_2|X_b)+{\mu_1}I(X_a;Y_1|X_2)\\
&+{\mu_2}I(X_a;Y_2|X_1), \\
R_{a2} = &{\mu_{14}}I(X_a;Y_2|X_b)+{\mu_{13}}I(X_1;Y_b|X_2)\\
&+{\mu_2}(I(X_a;Y_2)+I(X_1;Y_b)),\\
R_{a3} = &{\mu_{14}}I(X_a;Y_1|X_b)+{\mu_{13}}I(X_2;Y_b|X_1)\\
&+{\mu_1}(I(X_a;Y_1)+I(X_2;Y_b)),\\
R_{a4} = &{\mu_{13}}I(X_1,X_2;Y_b)+{\mu_1}I(X_2;Y_b)\\
&+{\mu_2}I(X_1;Y_b), 
\end{split}
\end{equation*} 
\begin{equation*}
\mbox{and~~} R_b \leq  \min{\lbrace R_{b1},R_{b2},R_{b3},R_{b4} \rbrace},
\end{equation*}
{\allowdisplaybreaks\begin{equation*}
\begin{split}
\mbox{where~~} R_{b1} = &{\mu_{14}}I(X_b;Y_1,Y_2|X_a)+{\mu_5}I(X_b;Y_1|X_2)\\
&+{\mu_6}I(X_b;Y_2|X_1), \\
R_{b2} = &{\mu_{14}}I(X_b;Y_2|X_a)+{\mu_{13}}I(X_1;Y_a|X_2)\\
&+{\mu_6}(I(X_b;Y_2)+I(X_1;Y_a)),\\
R_{b3} = &{\mu_{14}}I(X_b;Y_1|X_a)+{\mu_{13}}I(X_2;Y_a|X_1)\\
&+{\mu_5}(I(X_b;Y_1)+I(X_2;Y_b)),\\
R_{b4} = &{\mu_{13}}I(X_1,X_2;Y_a)+{\mu_5}I(X_2;Y_ba)\\
&+{\mu_6}I(X_1;Y_a). 
\end{split}
\end{equation*}} 
The mutual information terms in the above equations can be further bounded resulting in (\ref{outerbound}). For example, $I(X_b;Y_1,Y_2|X_a) \le C(\gamma_{b1} + \gamma_{b2})$ and $I(X_1,X_2;Y_b) \le C((\sqrt{\gamma_{b1}}+\sqrt{\gamma_{b2}})^2)$ \cite{khojastepour2003capacity,allerton12}. 
\end{proof}
\end{theorem}
The boundary of the capacity region in Theorem 1 can be computed by
solving the following linear program for each $k \ge 0$:
$\displaystyle{\max_{R_a,\{\mu_i\}} R_a,}$ subject to $R_a = kR_b$ and
the constraints in (\ref{outerbound}).
 
\section{Proposed Protocols and Achievable Rates}
\subsection{Two-way MDF protocol}
In \cite{BagMotKha10}, a multihopping-decode-and-forward (MDF)
protocol is proposed for one way communication in the diamond
channel. A simple two-way protocol can use the same MDF protocol for
both flows ($A$ to $B$ and $B$ to $A$) in a time-sharing manner. Thus,
states 1-4 in Fig. \ref{states} will be used for communication from
$A$ to $B$, and states 5-8 will be used for communication from $B$ to
$A$. We call this protocol that uses states 1-8 and a
decode-and-forward (DF) strategy in each state as the two-way MDF
protocol.

The maximum achievable rate for the one-way MDF protocol can be
computed as in \cite{BagMotKha10}. Suppose this rate for communication
from $A$ to $B$ is $R_{a-mdf}$ and for communication from $B$ to $A$
is $R_{b-mdf}$, then the achievable rate region for the two-way MDF
protocol is the triangular region enclosed by the three straight
lines: (1) $R_a = 0$, (2) $R_b = 0$, and (3) the line joining $(0,
R_{b-mdf})$ and $(R_{a-mdf}, 0)$.

\subsection{CF-CMAC protocol}
In this protocol, states 9, 10, and 13 are used. States 9 and 10 are
multiple access channels (MACs), in which both the nodes $A$ and $B$
transmit to one of the relays. This protocol employs a compute and forward
strategy at the relays, making use of doubly nested lattice codes following
\cite{NamChuLee10},\cite{interfere} to decode
the sum of the messages received from $A$ and $B$, instead of decoding
the individual messages. This sum is forwarded to the two nodes in
state 13. State 13 is a Compound MAC (CMAC), in
which both the relays simultaneously transmit to both nodes $A$ and
$B$. In a CMAC, each transmitter sends one message that should be decoded
at both the receivers. 

\subsubsection*{Doubly Nested Lattice codes}
A nested lattice code $\mathcal{L}$ is the set of all points of a fine
lattice $\Lambda$ that are within the fundamental Voronoi region
$\nu_1$ of a coarse lattice $\Lambda_1$, i.e., $\mathcal{L} =
\lbrace\Lambda \cap {\nu_1}\rbrace$. In \cite{NamChuLee10}, it was shown
that, for every $P_1 \geq P_2 \geq 0$, there exist a sequence of
$n$-dimensional lattices ${\Lambda_1}^n \subseteq {\Lambda_2}^n
\subseteq {\Lambda}^n$, as $n\to \infty$, with second moments
$\sigma^2({\Lambda_1}^n)\to P_1$
and $\sigma^2({\Lambda_2}^n)\to P_2$ such that the rate of the
nested lattice code $\mathcal{L}_2 = \lbrace{\Lambda^n} \cap
\nu_2\rbrace$ associated with the lattice partition
${\Lambda^n}/{{\Lambda_2}^n}$, as $n \to \infty$, approaches
\begin{equation}\label{R2}
R(\mathcal{L}_2) = \frac{1}{n}\log{|\mathcal{L}_2|} = \frac{1}{n}\log{\frac{\mbox{Vol}(\nu_2)}{\mbox{Vol}(\nu)}},
\end{equation}  
while the coding rate of the nested lattice code $\mathcal{L}_1 = \lbrace{\Lambda^n} \cap \nu_1\rbrace$ associated with the lattice partition ${\Lambda^n}/{{\Lambda_1}^n}$, as $n \to \infty$, approaches
\begin{equation}\label{R1}
\begin{split}
R(\mathcal{L}_1) =& \frac{1}{n}\log{|\mathcal{L}_1|} = \frac{1}{n}\log{\frac{\mbox{Vol}(\nu_1)}{\mbox{Vol}(\nu)}} = R(\mathcal{L}_2) + \frac{1}{2}\log{\frac{P_1}{P_2}}.
\end{split}
\end{equation} 
Using doubly nested lattice codes, an achievable rate region has been derived in \cite{NamChuLee10} for
the case of the Gaussian two-way relay channel. We use the same
design for the states 9 and 10 of the CF-CMAC protocol. We provide a
brief description for completeness, and refer to \cite{NamChuLee10} for
more details.

\subsubsection*{State 9} 
Consider an $n$-dimensional doubly nested lattice code $\mathcal{L}_1$, $\mathcal{L}_2$
with second moments $\sigma^2({\Lambda_1}^n)\to\gamma_{a1}$ and
$\sigma^2({\Lambda_2}^n)\to\gamma_{b1}$ assuming $\gamma_{a1}\ge \gamma_{b1}$. Node $A$ chooses a message $w_1
\in \mathcal{L}_1$, while $B$ chooses a message
$w_2 \in \mathcal{L}_2$. The nodes $A$ and $B$ transmit
\begin{equation}
X_1 = \frac{1}{h_{a1}}[(w_1+u_1)\mbox{mod } \Lambda_1], \;
X_2 = \frac{1}{h_{b1}}[(w_2+u_2)\mbox{mod } \Lambda_2],
\end{equation}
respectively, where $u_i$ are random dithers uniformly distributed in
the Voronoi regions of $\Lambda_1^n$ and $\Lambda_2^n$. Note that the
transmit signals $X_1$ and $X_2$ are pre-divided by the respective
channel gains to ensure that a noisy version of the sum $w_1+w_2$ is received at the relay. The relay $R_1$ receives
\begin{equation}
Y_{R_1} = h_{a1}X_1 + h_{b1}X_2 + Z_{R_1},
\end{equation}
multiplies $Y_{R_1}$ by the MMSE coefficient
$\displaystyle{\alpha = \frac{\gamma_{a1}+\gamma_{b1}}{1+\gamma_{a1}+\gamma_{b1}},}$
and subtracts the dithers to obtain
\begin{equation*}
{\hat{Y}_{R_1}} = (\alpha{Y_{R_1}} - u_1 - u_2) \mbox{mod } {\Lambda_1} 
=({T_1} + \hat{Z}_{R_1}) \mbox{mod } \Lambda_1,
\end{equation*}
where 
\begin{gather*}
T_1 
	=(w_1 - Q_{\Lambda_1}(w_1 + u_1)+ w_2 - Q_{\Lambda_2}(w_2 + u_2)) \mbox{mod } \Lambda_1, \\
\hat{Z}_{R_1} = -(1-\alpha){h_{a1}{X_1}} - (1-\alpha){h_{b1}{X_2}} + \alpha{Z_{R_1}}.
\end{gather*}
Here, $Q_\Lambda(\cdot)$ denotes quantization to the nearest lattice
point and $\hat{Z}_{R_1}$ is the effective noise at the relay with variance 
$\displaystyle{\mbox{Var}(\hat{Z}_{R_1}) = \frac{({\gamma_{a1}+\gamma_{b1}}){N}}{1+\gamma_{a1}+\gamma_{b1}},}$
which results in higher effective SNR at the relay.

The relay $R_1$ attempts to decode $T_1$ by quantizing $\hat{Y}_{R_1}$
to the closest point in the fine lattice, and an error occurs if
$\hat{Z}_{R_1}$ is outside the Voronoi region. The probability of
error vanishes \cite{interfere} as $n\to\infty$ if the second moment satisfies
\begin{equation}\label{condition}
{\sigma}^2(\Lambda^n) > \mbox{Var}({\hat{Z}_{R_1}}).
\end{equation}
From \eqref{R2}, \eqref{R1}, \eqref{condition}, the rate constraints for State 9 can be written as follows.
\begin{equation}\label{state1init}
\begin{split}
&F_{a{r_1}}^{9} \leq {\mu_9}\left[{\frac{1}{2}\log\left({\frac{\gamma_{a1}}{\gamma_{a1}+\gamma_{b1}}}+{\gamma_{a1}}\right)}\right]^{+} \\
&F_{b{r_1}}^{9} \leq {\mu_9}\left[{\frac{1}{2}\log\left({\frac{\gamma_{b1}}{\gamma_{a1}+\gamma_{b1}}}+{\gamma_{b1}}\right)}\right]^{+},
\end{split}
\end{equation}
where $\mu_9$ is the fraction of time for state 9, and
$F_{tr}^{s}$ denotes the amount of information flow from $t$ to
$r$ in state $s$.

\subsubsection*{State 10}
State 10 follows a similar lattice coding scheme as State 9, possibly
using a different doubly nested lattice code $\mathcal{L}_3$,
$\mathcal{L}_4$ with second moments $\sigma^2({\Lambda_3}^n)\to\gamma_{a2}$ and
$\sigma^2({\Lambda_4}^n)\to\gamma_{b2}$ assuming $\gamma_{a2}\ge \gamma_{b2}$. 
Nodes $A$ and $B$ choose messages $w_3$, $w_4$ from $\mathcal{L}_3$, $\mathcal{L}_4$,
and the relay $R_2$ decodes 
\begin{equation}
T_2=(w_3 - Q_{\Lambda_3}(w_3 + u_3)+ w_4 - Q_{\Lambda_4}(w_4 + u_4)) \mbox{mod } \Lambda_3,
\end{equation}
where $u_i$ are random dithers. Proceeding as for state 9, the rate constraints for state 10 can be written as
\begin{equation}\label{state2init}
\begin{split}
&F_{a{r_2}}^{10} \leq {\mu_{10}}\left[{\frac{1}{2}\log\left({\frac{\gamma_{a2}}{\gamma_{a2}+\gamma_{b2}}}+{\gamma_{a2}}\right)}\right]^{+} \\
&F_{b{r_2}}^{10} \leq {\mu_{10}}\left[{\frac{1}{2}\log\left({\frac{\gamma_{b2}}{\gamma_{a2}+\gamma_{b2}}}+{\gamma_{b2}}\right)}\right]^{+}.
\end{split}
\end{equation}

\subsubsection*{State 13} 
In state 13, both the relays simultaneously transmit the decoded sum
of messages to both nodes $A$ and $B$. Both $A$ and $B$ are required
to decode the message sent from each relay. Relay $R_1$ generates a
Gaussian codebook $C_{R_1}$ consisting of $|\mathcal{L}_1|$ $n$-length
sequences, with each element being {\it i.i.d} and having a Gaussian distribution
$\mathcal{N}(0,P)$. We assume that the relays make no error in
decoding $T_1$ and $T_2$ in the first two states, which will be true
if the rate constraints described above are satisfied. Since $T_1$ is
uniformly distributed over $\mathcal{L}_1$, for every $T_1 = t_1 \in
\mathcal{L}_1$, the relay chooses to transmit a particular
$X_{R_1}(t_1) \in C_{R_1}$. Similarly, $R_2$ generates a random
Gaussian codebook $C_{R_2}$ consisting of $|\mathcal{L}_3|$ $n$-length
sequences, with each element being {\it i.i.d}
$\mathcal{N}(0,P)$ and for every $T_2 = t_2 \in \mathcal{L}_3$, the
relay broadcasts a particular $X_{R_2}(t_2) \in C_{R_2}$. This results
in MACs at nodes $A$ and $B$. Node $A$ receives
$Y_A = h_{a1}{X_{R_1}} + h_{a2}{X_{R_2}} + Z_A$
from which it decodes $X_{R_1}$ and $X_{R_2}$. 

Since there is a one-to-one correspondence between the elements of
$\mathcal{L}_1$ and $C_{R_1}$, $A$ can obtain $\hat{T}_1$ from
${X_{R_1}}$. Also, $A$ can obtain $\hat{T}_2$ from ${X_{R_2}}$ because of
the one-to-one correspondence between the elements of $\mathcal{L}_3$
and $C_{R_2}$. Similarly, $B$ can obtain $\tilde{T}_1$ and $\tilde{T}_2$
from the received vector
$Y_B = h_{b1}{X_{R_1}} + h_{b2}{X_{R_2}} + Z_B.$

Since $w_1$, $w_3$ are messages transmitted by $A$ to the relays in
states 9 and 10, node $A$ has {\it a priori} knowledge of them. This can be used as side-information for decoding $w_2$
and $w_4$. Using the knowledge of $w_1$, $A$ can decode
$w_2$ from $\hat{T}_1$ as
$\hat{w}_2 = [\hat{T}_1 - w_1] \mbox{mod } \Lambda_2.$
Using $w_3$, $A$ can decode $w_4$ from $\hat{T}_2$ as 
$\hat{w}_4 = [\hat{T}_2 - w_3] \mbox{mod } \Lambda_4.$
Similarly, using its {\it a priori} knowledge of $w_2$, $w_4$ and dithers, node $B$ can decode $w_1$ and $w_3$ from $T_1$ and $T_2$ as
\begin{equation}
\begin{split}
\hat{w}_1 = [\tilde{T}_1 - w_2 + Q_{\Lambda_2}(w_2+u_2)] \mbox{mod } \Lambda_1 \\
\hat{w}_3 = [\tilde{T}_2 - w_4 + Q_{\Lambda_4}(w_4+u_4)] \mbox{mod } \Lambda_3.
\end{split}
\end{equation}
 
The rate constraints for correct CMAC decoding in state 13 are \cite{ahlswede}:
\begin{equation}\label{state3_1}
F_{{r_1}a}^{13} \leq {\mu_{13}}C(\gamma_{a1}), \; F_{{r_1}b}^{13} \leq {\mu_{13}}C(\gamma_{b1}), 
\end{equation}
\begin{equation}\label{state3_2}
F_{{r_2}a}^{13} \leq {\mu_{13}}C(\gamma_{a2}), \; F_{{r_2}b}^{13} \leq {\mu_{13}}C(\gamma_{b2}), 
\end{equation}
 \begin{equation}\label{state3_3}
\begin{split}
&F_{{r_1}a}^{13} + F_{{r_2}a}^{13} \leq {\mu_{13}}C(\gamma_{a1} + \gamma_{a2}), \\
&F_{{r_1}b}^{13} + F_{{r_2}b}^{13} \leq {\mu_{13}}C(\gamma_{b1} +
\gamma_{b2}), \\
\end{split}
\end{equation}
where the last 2 constraints are the MAC constraints at $A$ and $B$ on
the sum rates. Thus, this rate region is an intersection of the 2 MAC
regions defined by the 2 receivers $A$ and $B$.

Equating the information received at a relay from one node to the information forwarded by it to the other node, gives us the following four flow constraints.
\begin{equation}\label{flow1}
\begin{split}
F_{a{r_1}}^{9} = F_{{r_1}b}^{13},\, F_{b{r_1}}^{9} = F_{{r_1}a}^{13},\, 
F_{a{r_2}}^{10} = F_{{r_2}b}^{13},\, F_{b{r_2}}^{10} = F_{{r_2}a}^{13}. 
\end{split}
\end{equation}
The rates of information transfer between the end nodes $A$ and $B$ in the two directions are
\begin{equation}
\begin{split}
R_a = F_{a{r_1}}^{9} + F_{a{r_2}}^{10},\;\; 
R_b = F_{b{r_1}}^{9} + F_{b{r_2}}^{10}. 
\end{split}
\end{equation}
In summary, the achievable rate region can be obtained by taking $R_b
= kR_a$ and solving the linear program
$\max_{\lbrace{\mu_i}\rbrace} R_b$ with the rate constraints for
states 9, 10 and 13, flow constraints \eqref{flow1} and
$\sum_{i} \mu_i = 1$, $\mu_i\ge0$ as constraints, for various values of $k$.

\subsection{CF-BC protocol}
The CF-BC protocol uses states 9-12. Only one relay is used in any
state. States 9 and 10 are used the same way as in the CF-CMAC
protocol and have the same rate constraints. However, the computed sum
of messages at each relay is forwarded to the 2 destinations in a time-shared
manner, i.e., relay $R_1$ transmits to $A$ and $B$ using the broadcast
state 11 and relay $R_2$ transmits to $A$ and $B$ using the broadcast
state 12. The rate constraints for states 11 and 12 are:
\begin{equation}\label{state11_1}
F_{{r_1}a}^{11} \leq {\mu_{11}}C(\gamma_{a1}), \; F_{{r_1}b}^{11} \leq {\mu_{11}}C(\gamma_{b1}), \\
\end{equation}
\begin{equation}\label{state12_2}
F_{{r_2}a}^{12} \leq {\mu_{12}}C(\gamma_{a2}), \; F_{{r_2}b}^{12} \leq {\mu_{12}}C(\gamma_{b2}). 
\end{equation}

\subsection{Time-sharing between CF-CMAC and two-way MDF}
The CF-CMAC protocol achieves some rate-pairs that the two-way MDF
protocol cannot achieve. Time-sharing between CF-CMAC and two-way MDF
can be used to achieve all convex combinations of rate-pairs achieved by
the two protocols. Such a protocol would used $8+3=11$ states. 

\section{Two-way Diamond Channel with Direct Source-Destination Link}
In this Section, we consider the diamond network with a direct link between nodes $A$ and $B$ with SNR $\gamma_{ab}=\frac{{h^2}_{ab}P}{N}$. 

\subsection{Outer Bound}
In Section \ref{outer-diamond}, 6 states were considered to
obtain the outer bound. In this network with direct source-destination
link, 10 states need to be considered for obtaining the outer
bound. They are states 1-8, 13 and 14. Note that the $A$-$B$ link
should also be included in states 1-8. In states 13 and 14, the direct
link does not play a role since nodes $A$ and $B$ are
both transmitters or both receivers.

\begin{theorem}
Any achievable rate pair $(R_a, R_b)$ for two-way communication over the diamond channel with direct source-destination link must satisfy the following inequalities for some $\{\mu_i\}$: 
{\allowdisplaybreaks
\begin{align}
R_a \leq &{\mu_{14}}C(\gamma_{a1}+\gamma_{a2}) + {\mu_{4}}C(\gamma_{ab}) + {\mu_2}C(\gamma_{a2}+\gamma_{ab}) \nonumber \\
&+ {\mu_1}C(\gamma_{a1}+\gamma_{ab}) + {\mu_3}C(\gamma_{a1}+\gamma_{a2}+\gamma_{ab}), \nonumber \\
R_a \leq &{\mu_{14}}C(\gamma_{a2}) + {\mu_{13}}C(\gamma_{b1}) + {\mu_4}C((\sqrt{\gamma_{b1}}+\sqrt{\gamma_{ab}})^2) \nonumber \\
&+ {\mu_2}(C(\gamma_{b1})+C(\gamma_{a2}+\gamma_{ab})) + {\mu_1}C(\gamma_{ab}) \nonumber \\
&+ {\mu_3}C(\gamma_{a2}+\gamma_{ab}), \nonumber \\
R_a \leq &{\mu_{14}}C(\gamma_{a1}) + {\mu_{13}}C(\gamma_{b2}) +  {\mu_4}C((\sqrt{\gamma_{b2}}+\sqrt{\gamma_{ab}})^2) \nonumber \\
&+ {\mu_2}C(\gamma_{ab}) +{\mu_1}(C(\gamma_{b2})+C(\gamma_{a1}+\gamma_{ab})) \nonumber \\
&+ {\mu_3}C(\gamma_{a1}+\gamma_{ab}), \nonumber \\
R_a \leq &{\mu_{13}}C((\sqrt{\gamma_{b1}}+\sqrt{\gamma_{b2}})^2) +{\mu_4}C\left(\left(\sqrt{\gamma_{b1}}+\sqrt{\gamma_{b2}}\right.\right. \nonumber \\ &\left.\left.+\sqrt{\gamma_{ab}}\right)^2\right.)+{\mu_2}C((\sqrt{\gamma_{b1}}+\sqrt{\gamma_{ab}})^2)  \nonumber \\ 
&+{\mu_1}C((\sqrt{\gamma_{b2}}+\sqrt{\gamma_{ab}})^2)+ {\mu_3}C(\gamma_{ab}), \nonumber \\
R_b \leq &{\mu_{14}}C(\gamma_{b1}+\gamma_{b2}) + {\mu_{8}}C(\gamma_{ab}) + {\mu_6}C(\gamma_{b2}+\gamma_{ab}) \nonumber \\
 &+ {\mu_5}C(\gamma_{b1}+\gamma_{ab}) + {\mu_{7}}C(\gamma_{b1}+\gamma_{b2}+\gamma_{ab}), \nonumber \\
R_b \leq &{\mu_{14}}C(\gamma_{b2}) + {\mu_{13}}C(\gamma_{a1}) + {\mu_{8}}C((\sqrt{\gamma_{a1}}+\sqrt{\gamma_{ab}})^2) \nonumber \\
&+ {\mu_6}(C(\gamma_{a1})+C(\gamma_{b2}+\gamma_{ab})) + {\mu_5}C(\gamma_{ab}) \nonumber \\
&+ {\mu_{7}}C(\gamma_{b2}+\gamma_{ab}), \nonumber \\
R_b \leq &{\mu_{14}}C(\gamma_{b1}) + {\mu_{13}}C(\gamma_{a2}) +  {\mu_{8}}C((\sqrt{\gamma_{a2}}+\sqrt{\gamma_{ab}})^2) \nonumber \\
&+ {\mu_6}C(\gamma_{ab}) +{\mu_5}(C(\gamma_{a2})+C(\gamma_{b1}+\gamma_{ab})) \nonumber \\
&+ {\mu_{7}}C(\gamma_{b1}+\gamma_{ab}), \nonumber \\
R_b \leq &{\mu_{13}}C((\sqrt{\gamma_{a1}}+\sqrt{\gamma_{a2}})^2) +{\mu_{8}}C\left(\left(\sqrt{\gamma_{a1}}+\sqrt{\gamma_{a2}}\right.\right.\nonumber \\ &\left.\left.+\sqrt{\gamma_{ab}})^2)\right.\right. +{\mu_6}C((\sqrt{\gamma_{a1}}+\sqrt{\gamma_{ab}})^2)  \nonumber \\
&+{\mu_5}C((\sqrt{\gamma_{a2}}+\sqrt{\gamma_{ab}})^2)+ {\mu_{7}}C(\gamma_{ab}), \nonumber \\
\sum_{i} & \mu_i = 1, \mu_i\ge0,
\label{outerbound2}
\end{align}}
where the fraction of time network is in state $i$ is denoted $\mu_i$.
\label{thm2}
\end{theorem}
\begin{proof}
The proof is similar to the proof for Theorem \ref{thm1}. The same 4 cuts are considered. However, 10 states are considered instead of the 6 states in Theorem \ref{thm1}. 
\end{proof}

\subsection{Achievable Rate Regions}
The CF-CMAC and CF-BC protocols can be used for this network as
well. Since nodes $A$ and $B$ are both transmitters or both receivers
in the states used in the CF-BC and CF-CMAC protocols, the direct link
does not affect the protocol and the achievable rate regions remain
the same as in the case without the direct link. 

The Cooperative Multiple Access Broadcast Channel (CoMABC) protocol
proposed in \cite{tian2012asymmetric} for two way relaying with one
relay makes use of the direct link. This protocol is a three state
protocol. The first state is a MAC where both end nodes $A$
and $B$ transmit to the relay. The second state is a broadcast from
the relay. If the link $A$ to relay is better than the one from $B$,
then $A$ may transmit more bits in the first state and $B$ may receive
at much lower rate in the second state. This is compensated using a
third co-operative state after $A$ finishes decoding in the broadcast
state. In the third state, $A$ and the relay together transmit to
$B$. $A$ may re-transmit some information to help $B$ in decoding, or
may choose to transmit altogether new information. As in the CF-BC
protocol, we now use the CoMABC protocol with either with $R_1$ or
with $R_2$ depending on the rate pair to be achieved. We call this
protocol the 2-relay CoMABC protocol.

Assuming that $\gamma_{a1} \ge \gamma_{b1}$ and $\gamma_{a2} \ge
\gamma_{b2}$, the 2-relay CoMABC protocol uses states 9, 11, 2 and 10,
12, 1. The achievable rate region is the closure of the set of all
points $(R_a,R_b)$ satisfying following constraints:
{\allowdisplaybreaks\begin{align}
R_a &= R_{a1} + R_{a2}, R_{b} = R_{b1} + R_{b2} \mbox{  where} \nonumber \\
R_{a1} &\leq \min\{{{\mu_9}R_{ar1}^* + \mu_2\mathcal{C}(\gamma_{ab})}, \mu_{11}\mathcal{C}(\gamma_{b1})+\mu_2\mathcal{C}(\gamma_{b1}+\gamma_{ab})\}, \nonumber \\
R_{b1} &\leq \min\{{{\mu_9}R_{br1}^*},{\mu_{11}\mathcal{C}(\gamma_{a1})}\},  \nonumber \\
R_{a2} &\leq \min\{{{\mu_{10}}R_{ar2}^* + \mu_{1}\mathcal{C}(\gamma_{ab})}, \mu_{12}\mathcal{C}(\gamma_{b2})+\mu_1\mathcal{C}(\gamma_{b2}+\gamma_{ab})\}, \nonumber \\
R_{b2} &\leq \min\{{{\mu_{10}}R_{br2}^*},{\mu_{12}\mathcal{C}(\gamma_{a2})}\},  \nonumber \\
R_{ari}^* &= \left[{C(\gamma_{ai} - \frac{\gamma_{bi}}{{\gamma_{ai}+\gamma_{bi}}})}\right]^{+},~~ i = 1, 2  \nonumber \\
R_{bri}^* &= \left[{C(\gamma_{bi} - \frac{\gamma_{ai}}{{\gamma_{ai}+\gamma_{bi}}})}\right]^{+},~~ i = 1, 2 
\end{align}}
where $[x]^+\stackrel{\Delta}{=} \max(x,0)$ and the fraction of time
the network is in state $i$ is denoted $\mu_i$. When $\gamma_{b1} \ge
\gamma_{a1}$, state 6 is used instead of state 2 and a similar region
can be obtained. Similarly, when $\gamma_{b2} \ge \gamma_{a2}$, state
5 is used instead of state 1 and a similar region can be obtained.

\section{Two-way Diamond Channel with Interfering Relays}
In this Section, we consider the diamond network with a link between nodes $R_1$ and $R_2$ with SNR $\gamma_{12}=\frac{{h^2}_{12}P}{N}$. 

\subsection{Outer Bound}
In this network with interfering relays, 10 states need to be
considered for obtaining the outer bound. They are states 1, 2, 5, 6,
9-14.
\begin{theorem}
Any achievable rate pair $(R_a, R_b)$ for two-way communication over the diamond channel with interfering relays must satisfy the following inequalities for some $\{\mu_i\}$: 
{\allowdisplaybreaks\begin{align}
R_a \leq &{\mu_{14}}C(\gamma_{a1}+\gamma_{a2}) + {\mu_{10}}C(\gamma_{a2}) + {\mu_2}C(\gamma_{a2}) \nonumber\\
&+ {\mu_9}C(\gamma_{a1}) + {\mu_1}C(\gamma_{a1}), \nonumber\\
R_a \leq &{\mu_{14}}C(\gamma_{a2}) + {\mu_{13}}C(\gamma_{b1}) + {\mu_{10}}C((\sqrt{\gamma_{a2}}+\sqrt{\gamma_{12}})^2) \nonumber\\
&+ {\mu_2}(C(\gamma_{a2})+C(\gamma_{b1}+\gamma_{12})) \nonumber\\
&+ {\mu_{11}}C(\gamma_{b1}+\gamma_{12})+ {\mu_{6}}C(\gamma_{12}), \nonumber\\
R_a \leq &{\mu_{14}}C(\gamma_{a1}) + {\mu_{13}}C(\gamma_{b2}) + {\mu_9}C((\sqrt{\gamma_{a1}}+\sqrt{\gamma_{12}})^2) \nonumber\\
&+ {\mu_1}(C(\gamma_{a1})+C(\gamma_{b2}+\gamma_{12})) \nonumber\\
&+ {\mu_{12}}C(\gamma_{b2}+\gamma_{12})+ {\mu_5}C(\gamma_{12}), \nonumber\\
R_a \leq &{\mu_{13}}C((\sqrt{\gamma_{b1}}+\sqrt{\gamma_{b2}})^2)+{\mu_2}C(\gamma_{b1}) +{\mu_1}C(\gamma_{b2}) \nonumber\\
&+{\mu_{11}}C(\gamma_{b1})+{\mu_{12}}C(\gamma_{b2}), \nonumber\\
R_b \leq &{\mu_{14}}C(\gamma_{b1}+\gamma_{b2}) + {\mu_{10}}C(\gamma_{b2}) + {\mu_9}C(\gamma_{b1}) \nonumber\\
&+ {\mu_5}C(\gamma_{b1}) + {\mu_{6}}C(\gamma_{b2}), \nonumber\\
R_b \leq &{\mu_{14}}C(\gamma_{b2}) + {\mu_{13}}C(\gamma_{a1}) + {\mu_{10}}C((\sqrt{\gamma_{b2}}+\sqrt{\gamma_{12}})^2) \nonumber\\
&+ {\mu_2}(C(\gamma_{12}) + {\mu_{11}}C(\gamma_{a1}+\gamma_{12}) \nonumber\\
&+ {\mu_{6}}(C(\gamma_{b2})+C(\gamma_{a1}+\gamma_{12})), \nonumber\\
R_b \leq &{\mu_{14}}C(\gamma_{b1}) + {\mu_{13}}C(\gamma_{a2}) + {\mu_9}C((\sqrt{\gamma_{b1}}+\sqrt{\gamma_{12}})^2) \nonumber\\
&+ {\mu_1}(C(\gamma_{12}) + {\mu_{5}}(C(\gamma_{b1})+C(\gamma_{a2}+\gamma_{12})) \nonumber\\
&+ {\mu_{12}}C(\gamma_{a2}+\gamma_{12}), \nonumber\\
R_b \leq &{\mu_{13}}C((\sqrt{\gamma_{a1}}+\sqrt{\gamma_{a2}})^2) +{\mu_{11}}C(\gamma_{a1})+{\mu_5}C(\gamma_{a2}) \nonumber\\
&+{\mu_{12}}C(\gamma_{a2})+{\mu_{6}}C(\gamma_{a1}), \nonumber\\
\sum_{i} & \mu_i = 1, \mu_i\ge0,
\label{outerbound3}
\end{align}}
where the fraction of time network is in state $i$ is denoted $\mu_i$.
\label{thm3}
\end{theorem}
\begin{proof}
The proof is similar to the proof for Theorem \ref{thm1}. The same 4 cuts are considered. However, 10 states are considered instead of the 6 states in Theorem \ref{thm1}.
\end{proof}

\subsection{Achievable Rate Regions}
The CF-CMAC and CF-BC protocols can be used even in this case. The
achievable rate region for these protocols remains the same as in the
case without the $R_1$-$R_2$ link, i.e., the CF-BC and CF-CMAC
protocols do not use the link between the two relays.

Another achievable rate region for this channel using the $R_1$-$R_2$
link can be obtained using States 1, 2, 5 and 6 by time-sharing the
one-way alternating path relaying protocol in \cite{altrel} between
the two flows. We denote this protocol the 2-way AR-DF
protocol. Consider the flow from $A$ to $B$. In State 2, $A$ transmits
to $R_2$ while $R_1$ transmits to $B$. $R_2$ decodes the message from
$A$, and in addition decodes the message from $R_1$ too, considering
it as data to be forwarded to $B$ in the next state. Thus, $R_2$ acts as
a relay for both $A$ and $R_1$. In state 1, $A$ transmits new
information to $R_1$, $R_2$ transmits to $B$. The message transmitted
by $R_2$ will be a combination of the messages received from $A$ and
$R_2$ and the message to be sent to $R_1$. States 5 and 6 act in the
same way for the flow in the opposite direction, $R_b$.

The achievable rate region for the 2-way AR-DF protocol is the closure of
the set of all points $(R_a,R_b)$ satisfying following constraints
\cite{altrel}:
{\allowdisplaybreaks\begin{align}
R_a =&R_{a1} + R_{a2}, R_{b} = R_{b1} + R_{b2} \mbox{  where} \nonumber \\
R_{a1} &\leq {\mu_2}\mathcal{C}({\alpha_1}\gamma_{a2}), \; R_{a2} \leq {\mu_1}\mathcal{C}({\alpha_2}\gamma_{a1}), \nonumber \\
R_{a1} &\leq {\mu_2}\mathcal{C}\left( \frac{\beta_1\gamma_{b1}}{1+(1-\beta_1)\gamma_{b1}}\right)  + {\mu_1}\mathcal{C}((1-\beta_2)\gamma_{b2}), \nonumber \\
R_{a2} &\leq {\mu_1}\mathcal{C}\left( \frac{\beta_2\gamma_{b2}}{1+(1-\beta_2)\gamma_{b2}}\right)  + {\mu_2}\mathcal{C}((1-\beta_1)\gamma_{b1}), \nonumber \\
R_{a1} + R_{a2} &\leq {\mu_2}\mathcal{C}\left(\gamma_{a2}+(1-{\beta_1})\gamma_{12}\right. \nonumber\\
&\left.+2\sqrt{(1-\alpha_1)(1-\beta_1)\gamma_{a2}\gamma_{12}}\right), \nonumber \\
R_{a1} + R_{a2} &\leq {\mu_1}\mathcal{C}\left(\gamma_{a1}+(1-{\beta_2})\gamma_{12}\right. \nonumber\\
&\left.+2\sqrt{(1-\alpha_2)(1-\beta_2)\gamma_{a1}\gamma_{12}}\right), \nonumber\\
R_{b1} &\leq {\mu_5}\mathcal{C}({\alpha_3}\gamma_{b2}), \; R_{b2} \leq {\mu_6}\mathcal{C}({\alpha_4}\gamma_{b1}),\nonumber\\
R_{b1} &\leq {\mu_5}\mathcal{C}\left( \frac{\beta_3\gamma_{a1}}{1+(1-\beta_3)\gamma_{a1}}\right)  + {\mu_6}\mathcal{C}((1-\beta_4)\gamma_{a2}),\nonumber\\
R_{b2} &\leq {\mu_6}\mathcal{C}\left( \frac{\beta_4\gamma_{a2}}{1+(1-\beta_4)\gamma_{a2}}\right) + {\mu_5}\mathcal{C}((1-\beta_3)\gamma_{a1}),\nonumber\\
R_{b1} + R_{b2} &\leq {\mu_5}\mathcal{C}\left(\gamma_{b2}+(1-{\beta_3})\gamma_{12}\right.\nonumber\\
&\left.+2\sqrt{(1-\alpha_3)(1-\beta_3)\gamma_{b2}\gamma_{12}}\right), \nonumber\\
R_{b1} + R_{b2} &\leq {\mu_6}\mathcal{C}\left(\gamma_{b1}+(1-{\beta_4})\gamma_{12}\right.\nonumber\\
&\left.+2\sqrt{(1-\alpha_4)(1-\beta_4)\gamma_{b1}\gamma_{12}}\right),
\end{align}}
where the fraction of time in state $i$ is denoted $\mu_i$.

\section{Numerical Results and Comparisons}
\subsection{Diamond Channel}
In this section, we compare the achievable rate regions of the two-way MDF, CF-BC, and CF-CMAC protocols and also compare with the capacity region outer bound. Numerical results are shown for three different channels: I, II, and III.
\begin{figure}[!ht]
\centering
\includegraphics[width=3in]{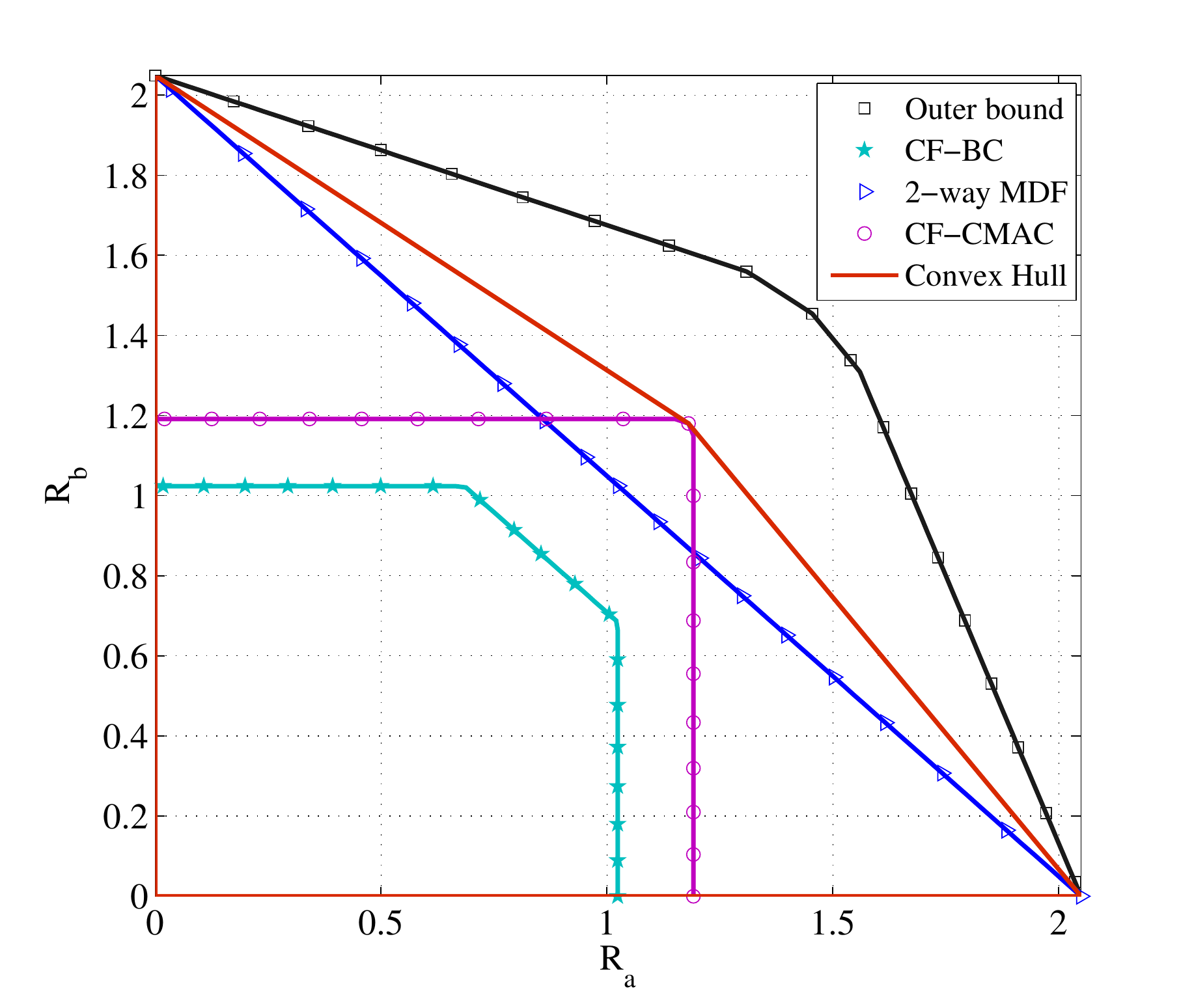}
\caption{Comparison of rate regions -- Channel I: $\gamma_{a1} = 15$ dB, $\gamma_{b1} = 10$ dB, $\gamma_{a2} = 10$ dB, $\gamma_{b2} = 15$ dB.}
\label{figure4}
\end{figure}

Figure \ref{figure4} shows the comparison of rate regions for channel I. It has been shown in \cite{BagMotKha10} that under this channel condition the one-way MDF protocol achieves capacity. Therefore, the two-way MDF region meets the outer bound on the axes where it corresponds to one-way MDF. However, away from the two axes, the two-way MDF has a significant gap from the outerbound. The CF-BC protocol, which uses one relay at a time, is not able to provide any improvement in this scenario. However, the proposed CF-CMAC protocol is able to achieve rate-pairs outside the rate region of the two-way MDF protocol. The convex combination of the CF-CMAC and two-way MDF protocol rate regions is also shown (labelled ``convex hull''). Any point in this convex hull can be achieved by time-sharing the CF-CMAC and two-way MDF protocols.  

\begin{figure}[!ht]
\centering
\includegraphics[width=3in]{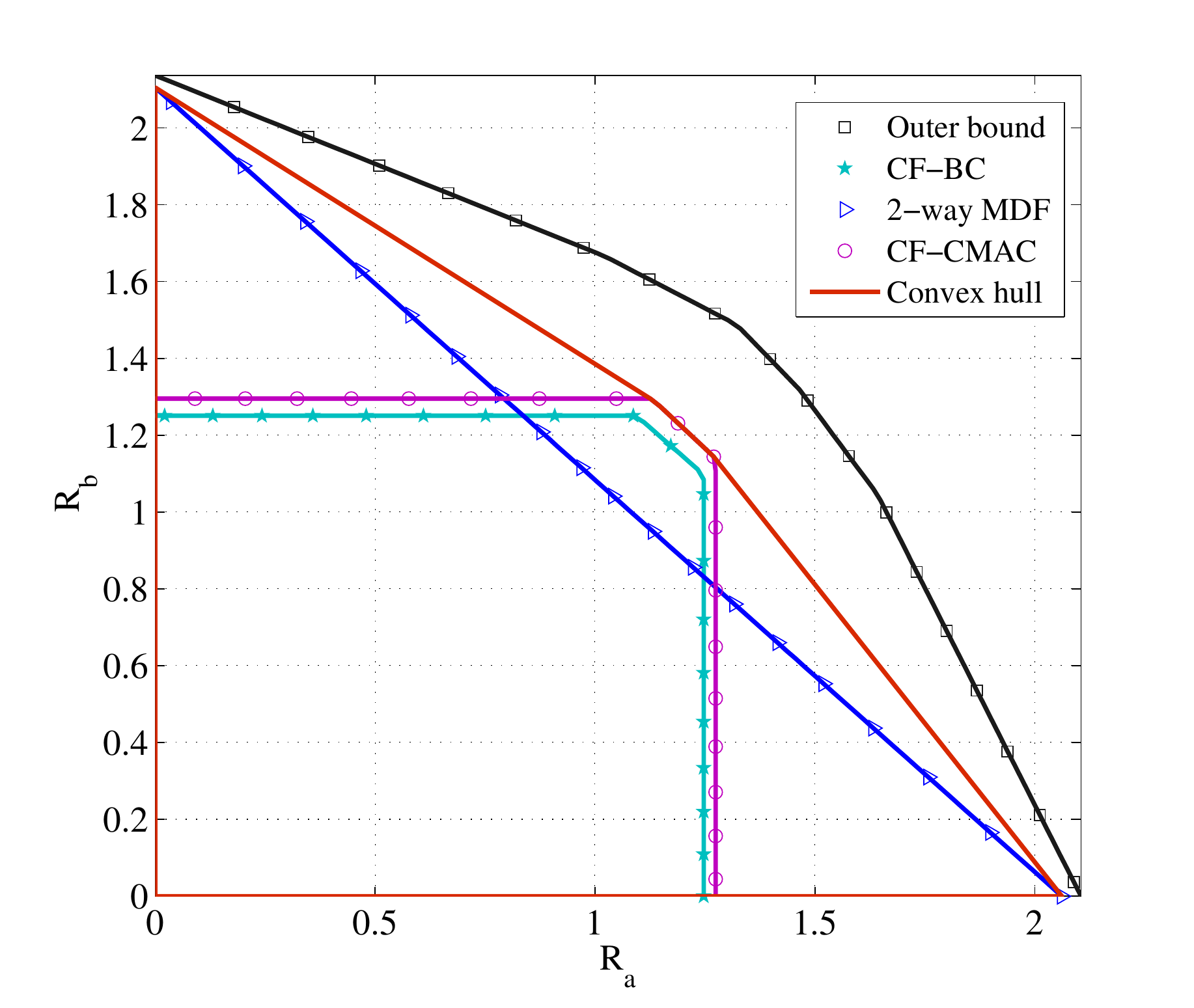}
\caption{Comparison of rate regions -- Channel II: $\gamma_{a1} = 10$ dB, $\gamma_{b1} = 12$ dB, $\gamma_{a2} = 14$ dB, $\gamma_{b2} = 16$ dB.}
\label{figure2}
\end{figure}
Figure \ref{figure2} shows the comparison of rate regions for channel II. In this scenario, both the CF-BC and CF-CMAC protocols achieve rate-pairs outside the two-way MDF rate region. Even CF-BC, which selects one relay and uses compute-and-forward, provides gains in this case. 

Figure \ref{figure5} shows the comparison of rate regions for channel III. In this scenario, the links to relay $R_1$ are significantly better than the links to $R_2$ in terms of SNR. Therefore, this scenario is closer to a one-relay system. Both the CF-BC and CF-CMAC protocols achieve rate-pairs outside the two-way MDF rate region and close to the outer bound. 
\begin{figure}[!ht]
\centering
\includegraphics[width=3in]{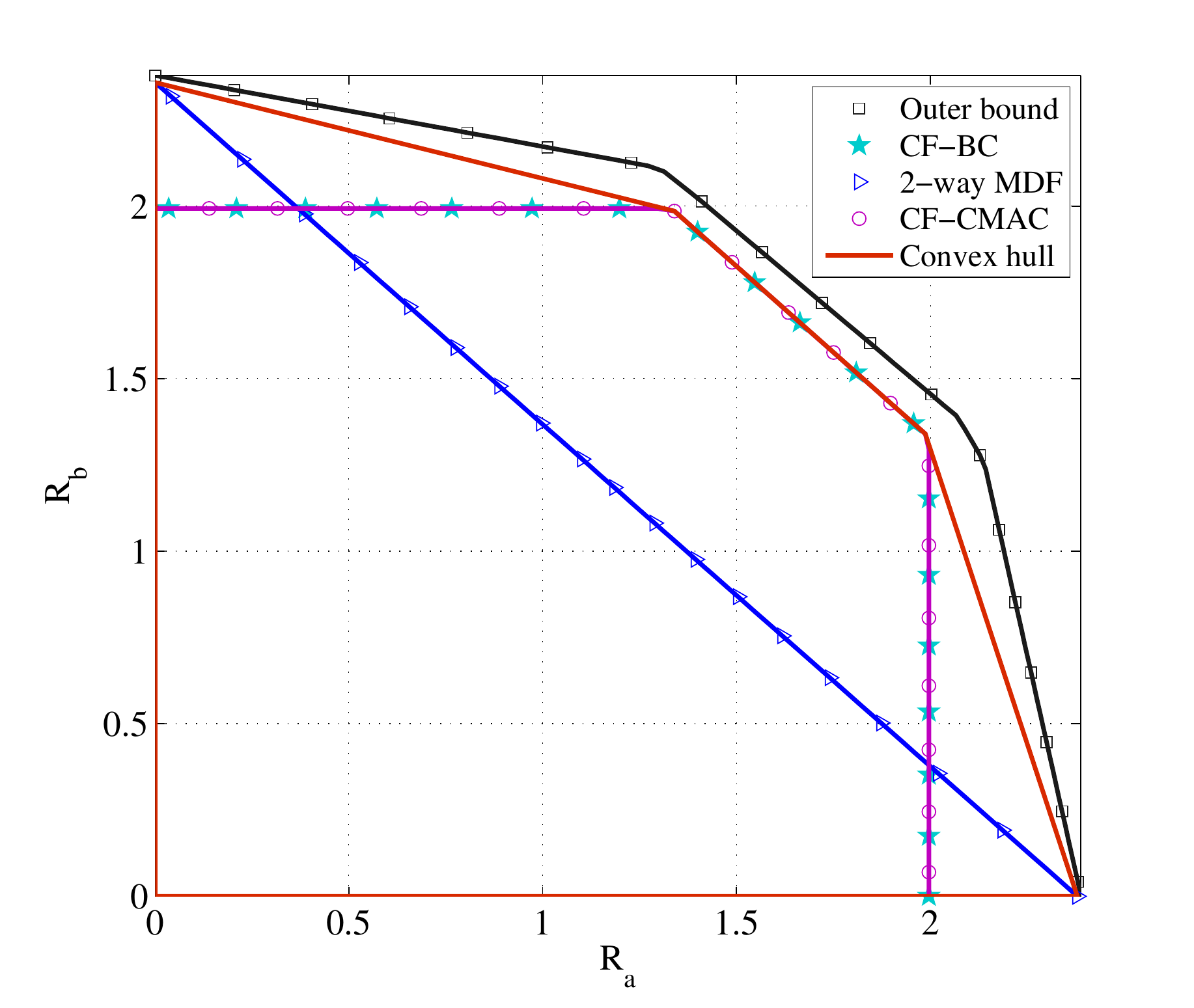}
\caption{Comparison of rate regions -- Channel III: $\gamma_{a1} = 30$ dB, $\gamma_{b1} = 20$ dB, $\gamma_{a2} = 3$ dB, $\gamma_{b2} = 4$ dB.}
\label{figure5}
\end{figure}

In general, the one-way MDF protocol is close to capacity for the one-way diamond channel. Therefore, by time-sharing, one can obtain rates close to capacity near the axes. However, when the desired two-way rates are nearly equal, time-sharing of one-way protocols is far from optimal, and the proposed CF-CMAC and CF-BC protocols achieve much better rates.

\subsection{Diamond Channel with Direct Link}
\begin{figure}[!ht]
\centering
\includegraphics[width=3.2in]{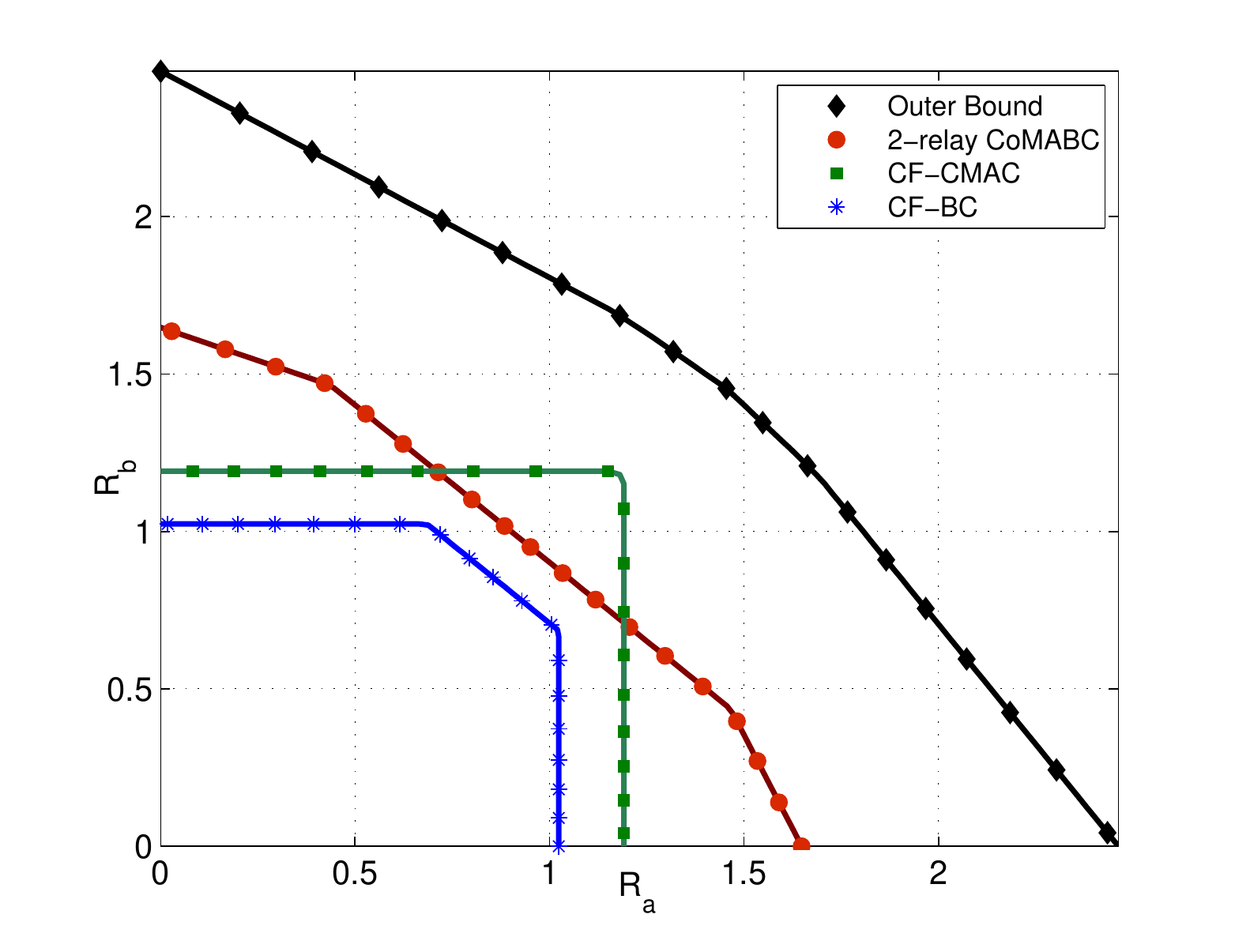}
\caption{Comparison of rate regions: $\gamma_{a1} = 15$ dB, $\gamma_{b1} = 10$ dB, $\gamma_{a2} = 10$ dB, $\gamma_{b2} = 15$ dB, $\gamma_{ab} = 8$ dB.}
\label{ar1}
\end{figure}
Figure \ref{ar1} compares the achievable rate regions of the CF-BC,
2-relay CoMABC protocol, and CF-CMAC protocols with the outer bound.
The 2-relay CoMABC protocol always achieves a larger rate region
than the CF-BC protocol since it uses the direct link as well. The
CF-CMAC can achieve some rate pairs that the 2-relay CoMABC cannot
achieve even without using the direct link. This is because both
relays transmit simultaneously in state 13 of the CF-CMAC
protocol. Mixing of flows using the CF-CMAC protocol performs well for
symmetric rates, while the 2-relay CoMABC protocol using the direct
link performs better for asymmetric rates.

\subsection{Diamond Channel with Interfering Relays}
\begin{figure}[!ht]
\centering
\includegraphics[width=3.2in]{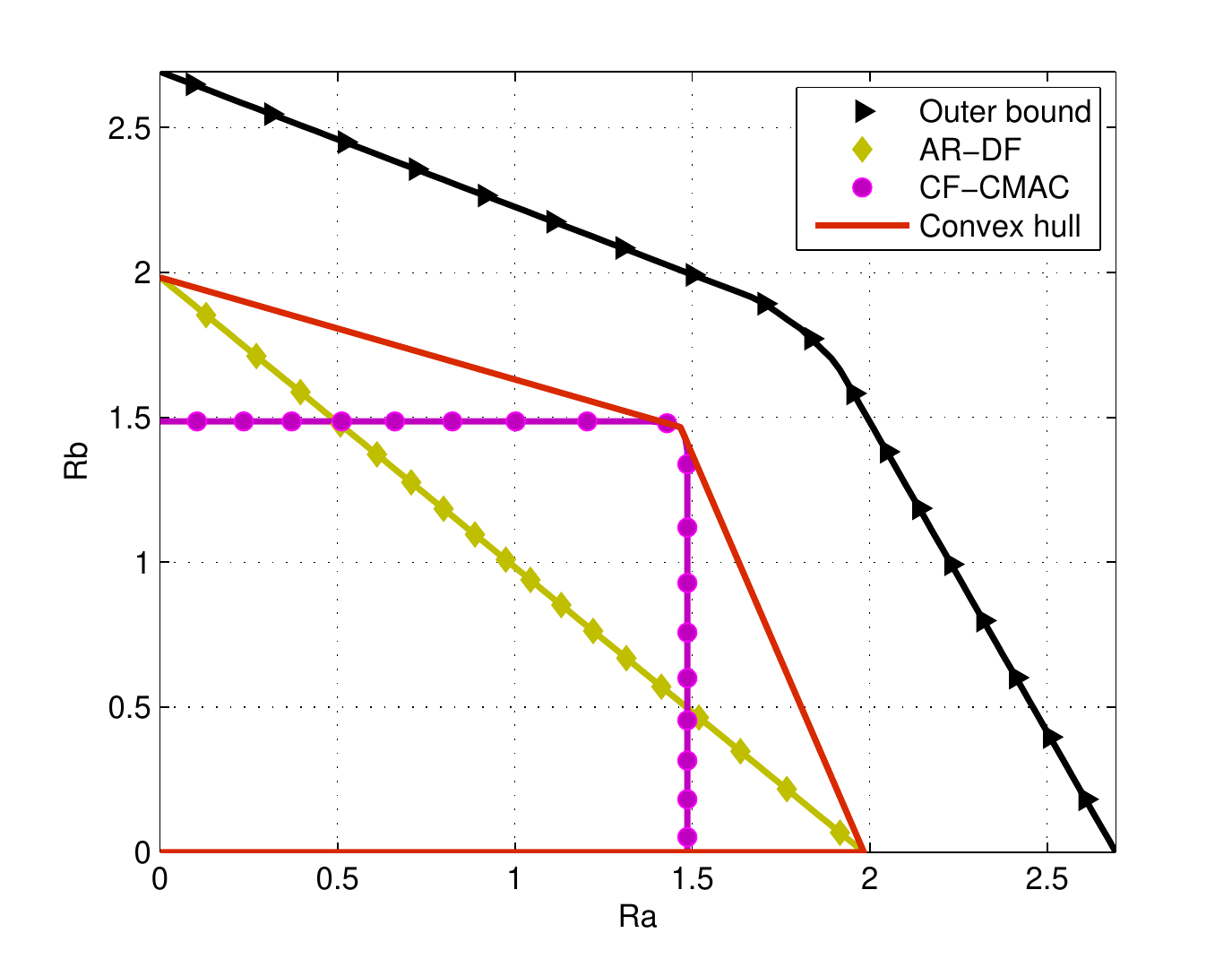}
\caption{Comparison of rate regions: $\gamma_{a1} = 20$ dB, $\gamma_{b1} = 10$ dB, $\gamma_{a2} = 10$ dB, $\gamma_{b2} = 20$ dB, $\gamma_{12} = 20$ dB.}
\label{ar2}
\end{figure}
Figure \ref{ar2} shows the achievable rate regions of the 2-way
AR-DF and CF-CMAC protocols with the outer bound. The convex hull of
the rates achieved by the CF-CMAC and 2-way AR-DF protocols is also
shown. Again, mixing of flows using the CF-CMAC protocol performs well
for symmetric rates, while the 2-way AR-DF protocol performs better
for asymmetric rates. The AR-DF protocol could be further improved
using the other one-way states 3, 4, 7, and 8.

\section{Conclusions}
In this paper, we considered the Gaussian two-way diamond channel. We derived
an outer bound for the capacity region and proposed relaying protocols
based on the compute-and-forward technique. The proposed CF-CMAC
protocol achieves rate-pairs that two-way protocols based on
time-sharing or relay selection cannot achieve. There is still significant 
room for improving the achievable rates or the outer bound in several interesting channel conditions. The 
interference channel state (state 14 in Fig. \ref{states}) will probably 
play a vital role in closing the gap further.

We also considered the diamond channel with direct source-destination
link and the diamond channel with interfering relays. We derived outer
bounds for the capacity region in each case. We developed the 2-relay
CoMABC protocol for the diamond channel with direct link and compared
the achievable rate region with the CF-BC and CF-CMAC protocols and
the outer bound. We also compared the 2-way AR-DF protocol with the
CF-CMAC protocol and the outer bound for the diamond channel with
interfering relays. As in the diamond channel, mixing the flows using
compute and forward strategies is found to be important for symmetric
rates.

\bibliographystyle{IEEEtran} 
\bibliography{IEEEabrv,master}

\begin{thebibliography}{10}
\providecommand{\url}[1]{#1}
\csname url@samestyle\endcsname
\providecommand{\newblock}{\relax}
\providecommand{\bibinfo}[2]{#2}
\providecommand{\BIBentrySTDinterwordspacing}{\spaceskip=0pt\relax}
\providecommand{\BIBentryALTinterwordstretchfactor}{4}
\providecommand{\BIBentryALTinterwordspacing}{\spaceskip=\fontdimen2\font plus
\BIBentryALTinterwordstretchfactor\fontdimen3\font minus
  \fontdimen4\font\relax}
\providecommand{\BIBforeignlanguage}[2]{{%
\expandafter\ifx\csname l@#1\endcsname\relax
\typeout{** WARNING: IEEEtran.bst: No hyphenation pattern has been}%
\typeout{** loaded for the language `#1'. Using the pattern for}%
\typeout{** the default language instead.}%
\else
\language=\csname l@#1\endcsname
\fi
#2}}
\providecommand{\BIBdecl}{\relax}
\BIBdecl

\bibitem{SchienDiamond}
B.~Schein and R.~Gallager, ``The {G}aussian parallel relay network,'' in
  \emph{Proceedings of the IEEE International Symposium on Information Theory},
  2000, p.~22.

\bibitem{BagMotKha10}
H.~Bagheri, A.~Motahari, and A.~Khandani, ``On the capacity of the half-duplex
  diamond channel,'' in \emph{Proceedings of the IEEE International Symposium
  on Information Theory}, June 2010, pp. 649--653.

\bibitem{PopYom07}
P.~Popovski and H.~Yomo, ``Physical network coding in two-way wireless relay
  channels,'' in \emph{IEEE ICC 2007}, 2007, pp. 707--712.

\bibitem{kim2008performance}
S.~Kim, P.~Mitran, and V.~Tarokh, ``Performance bounds for bidirectional coded
  cooperation protocols,'' \emph{IEEE Transactions on Information Theory},
  vol.~54, no.~11, pp. 5235--5241, Nov. 2008.

\bibitem{kim2011achievable}
S.~Kim, N.~Devroye, P.~Mitran, and V.~Tarokh, ``Achievable rate regions and
  performance comparison of half duplex bi-directional relaying protocols,''
  \emph{IEEE Transactions on Information Theory}, vol.~57, no.~10, pp.
  6405--6418, Oct. 2011.

\bibitem{NamChuLee10}
W.~Nam, S-Y.Chung, and Y.~H. Lee, ``Capacity of the {G}aussian two-way relay
  channel to within {1/2} bit,'' \emph{IEEE Transactions on Information
  Theory}, vol.~56, pp. 5488--5494, Nov. 2010.

\bibitem{WilNar10}
M.~Wilson, K.~Narayanan, H.~Pfister, and A.~Sprintson, ``Joint physical layer
  coding and network coding for bidirectional relaying,'' \emph{IEEE Trans. on
  Information Theory}, vol.~56, no.~11, pp. 5641--5654, Nov. 2010.

\bibitem{NazerGastpar}
B.~Nazer and M.~Gastpar, ``Compute-and-forward: Harnessing interference through
  structured codes,'' \emph{IEEE Transactions on Information Theory}, vol.~57,
  no.~10, pp. 6463--6486, Oct. 2011.

\bibitem{GonYueWan11}
C.~Gong, G.~Yue, and X.~Wang, ``A transmission protocol for a cognitive
  bidirectional shared relay system,'' \emph{IEEE Journal of Selected Topics in
  Signal Processing}, vol.~5, no.~1, pp. 160--170, Feb. 2011.

\bibitem{tian2012asymmetric}
Y.~Tian, D.~Wu, C.~Yang, and A.~Molisch, ``Asymmetric two-way relay with doubly
  nested lattice codes,'' \emph{IEEE Transactions on Wireless Communications},
  vol.~11, no.~2, pp. 694--702, Feb. 2012.

\bibitem{allerton12}
I.~Ashar, V.~Prathyusha, S.~Bhashyam, and A.~Thangaraj, ``Outer bounds for the
  capacity region of a {G}aussian two-way relay channel,'' in \emph{Allerton
  Conference on Communication, Control, and Computing}, Monticello, IL, USA,
  Oct. 2012.

\bibitem{khojastepour2003capacity}
M.~Khojastepour, A.~Sabharwal, and B.~Aazhang, ``On capacity of
  {G}aussian`cheap'relay channel,'' in \emph{IEEE Global Telecommunications
  Conference, 2003.}, vol.~3.\hskip 1em plus 0.5em minus 0.4em\relax IEEE,
  2003, pp. 1776--1780.

\bibitem{interfere}
W.~Nam, S.-Y. Chung, and Y.~Lee, ``Nested lattice codes for {G}aussian relay
  networks with interference,'' \emph{IEEE Transactions on Information Theory},
  vol.~57, no.~12, pp. 7733--7745, Dec. 2011.

\bibitem{ahlswede}
R.~Ahlswede, ``The capacity region of a channel with two senders and two
  receivers,'' \emph{Annals of Probability}, vol.~2, no.~5, pp. 805--814, Oct.
  1974.

\bibitem{altrel}
W.~Chang, S.-Y. Chung, and Y.~H. Lee, ``{Capacity Bounds for Alternating
  Two-Path Relay Channels},'' in \emph{Allerton Conference on Communication,
  Control, and Computing}, Monticello, IL, USA, Sep. 2007.

\end{thebibliography}

\end{document}